\documentclass[a4paper,11pt]{article}

\usepackage[utf8]{inputenc} 
\usepackage[UKenglish]{babel} 
\usepackage{comment}
\usepackage{amsmath, amssymb, amsthm} 
\usepackage{hyperref} 
\usepackage{cleveref} 
\usepackage{thm-restate} 
\usepackage{complexity} 

\usepackage[margin=1in]{geometry}

\title{On the Hardness of Order Finding and Equivalence Testing for ROABPs}

\author{
    C. Ramya\\
    \small The Institute of Mathematical Sciences, Chennai, India\\
    \small\texttt{ramyac@imsc.res.in}
    \and
    Pratik Shastri\\
    \small The Institute of Mathematical Sciences, Chennai, India\\
    \small\texttt{pratiks@imsc.res.in}
}

\date{} 

\newcommand{\F}{\mathbb{F}}
\newcommand{\GL}{\mathrm{GL}}

\newcommand{\cutwidth}{\mathrm{cutwidth}}
\newcommand{\cut}{\mathrm{cut}}
\newcommand{\rank}{\mathrm{Rank}}
\newcommand{\bfx}{\mathbf{x}}
\newcommand{\bfb}{\mathbf{b}}
\newcommand{\cutrank}{\mathrm{cut\mbox{-}rank}}
\newcommand{\rowidth}{\mathrm{RO\mbox{-}width}}

\newtheorem{theorem}{Theorem}[section]
\newtheorem{lemma}[theorem]{Lemma}
\newtheorem{claim}[theorem]{Claim}

\newtheorem{definition}[theorem]{Definition}
\newtheorem{problem}[theorem]{Problem}

\begin{document}

\maketitle

\begin{abstract}
The complexity of representing a polynomial by a Read-Once Oblivious Algebraic Branching Program (ROABP) is highly dependent on the chosen variable ordering. Bhargava et al. \cite{BDGT24} prove that finding the optimal ordering is NP-hard, and provide some evidence (based on the Small Set Expansion hypothesis) that it is also hard to approximate the optimal ROABP width. In another work, Baraskar et al. \cite{BDSS24} show that it is NP-hard to test whether a polynomial is in the $\GL_n$ orbit of a polynomial of sparsity at most $s$. Building upon these works, we show the following results: first, we prove that approximating the minimum ROABP width up to any constant factor is NP-hard, when the input is presented as a circuit. This removes the reliance on stronger conjectures in the previous work \cite{BDGT24}. Second, we show that testing if an input polynomial given in the sparse representation is in the affine $\GL_n$ orbit of a width-$w$ ROABP is NP-hard. Furthermore, we show that over fields of characteristic $0$, the problem is NP-hard even when the input polynomial is homogeneous. This provides the first NP-hardness results for membership testing for a dense subclass of polynomial sized algebraic branching programs (VBP). Finally, we locate the source of hardness for the order finding problem at the lowest possible non-trivial degree, proving that the problem is NP-hard even for quadratic forms.
\end{abstract}


\section{Introduction}

{\em Algebraic circuits} provide a powerful framework for understanding the complexity of computing multivariate polynomial over a field.  These are directed acyclic graphs whose in-degree 0 vertices are labeled by variables or field constants and internal vertices are either addition or multiplication gates. Several computational problems concerning polynomials such as identity testing, polynomial factoring, equivalence testing and reconstruction have been studied intensively for various structured subclasses of circuits.

In this article, we are interested in {\em Read-Once Oblivious Algebraic Branching Programs} (ROABPs for short), the algebraic analogs of ordered binary decision diagrams (OBDDs). Informally, an ROABP is layered DAG with a designated source and  sink vertex whose edges are labeled by univariate polynomials. More importantly, all edges in the same layer use the same variable and a variable occurs in exactly one layer. The polynomial computed by the ROABP is the sum over all source-to-sink paths of the product of edge weights along each path. In this work, we study the computational complexity of two problems concerning ROABPs: {\em order-finding} and {\em equivalence testing}. Before delving into the details of these two problems, we note that in algorithms concerning polynomials, a polynomial $f(x_1,\ldots ,x_n)$ of degree $d$ can be given as an input to an algorithm in one of the following standard representations: 
\begin{itemize}
    \item Dense Representation: a list (of size $\binom{n+d}{d}$) of coefficients for all possible monomials up to degree $d$. 
    \item Sparse Representation: a list of pairs, where each pair consists of a non-zero coefficient and its corresponding monomial (represented by its exponent vector).
    \item Circuit Representation: an algebraic circuit computing $f$.
    \item Black-Box Access: an oracle access to evaluations of $f$. 
\end{itemize}
It is important to note that the choice of representation can critically affect the complexity of a computational problem. The input representations above are ordered in increasing order of compactness. As a rule, problems considered in this paper become harder when the input is presented more compactly. Conversely, proving hardness becomes easier.


First, we begin with the order finding problem for ROABPs. It is easy to see that  with every ROABP computing a polynomial $f(x_1,\ldots ,x_n)$ we can associate a unique permutation $\sigma :[n] \rightarrow [n]$ which we call the {\em order} of the ROABP. The size of an ROABP, particularly its \emph{width}, which is the maximum number of vertices in any of its layers serves as a key measure of complexity. A crucial feature of ROABPs is that the width required to compute a polynomial $f$ is critically dependent on the order. A poor choice of ordering can lead to an exponential blow-up in the required width compared to an optimal one. For instance, the minimal widths of ROABPs computing $(x_1+y_1)(x_2+y_2)\cdots (x_n+y_n)$ in orders $(x_1,y_1,x_2,y_2,\ldots ,x_n,y_n)$ and $(x_1,\ldots ,x_n,y_1,\ldots ,y_n)$ are drastically different. This sensitivity gives rise to the natural computational question of order finding introduced in Bhargava et al.\ \cite{BDGT24}:

\begin{problem}[$\sf CktROWidth\mbox{-}d$] 
 [Decision version of order-finding]: Given a  polynomial $f\in\mathbb{F}[x_1, \ldots, x_n]$ of degree at most $d$ as an algebraic circuit and an integer $w$ in binary, decide whether there exists an ROABP for $f$ of width at most $w$, in some order.
\end{problem}

\begin{problem}[$\sf SearchCktROWidth\mbox{-}d$]
 [Search version of order-finding]: Given an algebraic circuit $C$ computing an $n$-variate polynomial $f$ of degree $\leq d$, find an order $\sigma \in S_n$ that minimizes the ROABP width for $f$.

\end{problem}

Depending on the input representation, other variations of these problems are defined analogously. For instance, DenseROWidth-d is the decision version of the problem when the input is provided in the dense representation.

Bhargava, Dutta, Ghosh, and Tengse \cite{BDGT24} study the complexity of the order-finding problem. In particular, they show that DenseROWidth-6 is $\NP$-hard. Their proof is an interesting reduction from the {\em cutwidth} problem for graphs. This is a {\em linear arrangement problem} in which given a graph the goal is to find an ordering (a.k.a linear arrangement) of the vertices  that minimizes the maximum number of edges between any prefix and the corresponding suffix in that ordering. It is known that the cutwidth problem is $\NP$ hard even for graphs with maximum degree $3$. Their reduction is parameter preserving: The graph is transformed into a polynomial in such a way that ROABP width of the polynomial is $2$ more than the cutwidth of the graph. The degree of the polynomial constructed is exactly $2\Delta$, where $\Delta$ denotes the maximum degree of the graph.

Furthermore,  Bhargava et al.\ \cite{BDGT24} also study the problem of approximating ROABP width. They provide two pieces of evidence that indicate hardness of approximation, both stemming from their reduction. First, as noted earlier, the reduction in \cite{BDGT24} is parameter preserving. Hence, hardness of approximation for the cutwidth problem directly translates to hardness of approximation for ROABP width. Austrin et al.\ \cite{APW12} showed that cutwidth is hard to approximate assuming the {\em small set expansion} (SSE) hypothesis thus implying hardness of approximating ROABP width assuming the SSE hypothesis, even when the input is given in the dense representation. Second, they develop a {\em tensoring technique} through which they show that a constant factor approximation for Search-CktROWidth would imply a PTAS for Search-CktROWidth, and therefore, by their reduction, for cutwidth. This is (potentially stronger) evidence against the existence of a constant factor approximation for Search-CktROWidth, since Ambühl et al.\ \cite{AMS11} show that for a problem related to cutwidth, a PTAS does not exist under assumptions weaker than SSE (but stronger than $\P\neq \NP$). 

In this work, we prove $\NP$-Hardness of approximating Search-CktROWidth up to any arbitrary constant factor.

\begin{theorem}[Inapproximability of Search-CktROWidth-d]
\label{thm:inapproximability}
    Let $\alpha\in \mathbb{R}$ be an arbitrary constant. Let $n,d\in \mathbb{N}$ be given as input, in unary. Let $f\in \mathbb{F}[x_1, \ldots, x_n]$ of degree $d$ be given as input, as an arithmetic circuit. Let $w\in\mathbb{N}$ also be give as input, in binary. It is $\NP$-hard to distinguish between the following two cases:
    \begin{enumerate}
        \item $f$ has an ROABP of width $\leq w$ in some order.
        \item Every ROABP for $f$ has width $>\alpha\cdot w$.
    \end{enumerate}
\end{theorem}

Our reduction is natural, direct, and again, from cutwidth. It works over any field. We use the power of the circuit representation to construct, from a graph $G$, a depth three circuit for a polynomial $f_G$ such that for every subset $S$ of vertices,  rank of the Nisan matrix of $f_G$ with respect to the set $S$ is exactly $2^{|\cut(S)|}$. For a subset of vertices $S$, $\cut(S)$ denotes the number of edges going out of $S$. See Section \ref{sec:prelims} for definition of Nisan matrices. This already gives NP-hardness of $2$-approximation. Then, we use the tensoring technique from \cite{BDGT24} to get hardness of $\alpha$-approximation for any constant $\alpha$. 

Next, we move on to the problem of testing equivalence to ROABPs. Polynomial equivalence testing is a well-studied problem in algebraic complexity theory: given two polynomials the goal is to decide if one is equivalent to the other via an invertible affine transformation of the variables. Several special cases of the polynomial equivalence problem have been studied time and again. In order to understand recent progress on this problem, we consider the notion of {\em orbits} of polynomial families. Let $\bfx$ denote $(x_1, \ldots, x_n)$. The {\em orbit} of an $n$-variate polynomial $f$ is the set of polynomials obtained from $f$ by applying an invertible affine transformation to the variables, i.e., $ {\sf orbit}(f) =  \{  f(A\mathbf{x} + \mathbf{b}) \mid A \in \GL_n(\F) \text{ and } \mathbf{b} \in \F^n \}$. For any class $C$ of polynomials, the orbit of $C$ is the union of orbits of polynomials in $C$. The orbits of the determinant and permanent polynomials are central to geometric complexity theory.

In the equivalence problem for a certain class $C$ of circuits we are given a polynomial $f$ (in some representation) and the goal is to determine if $f$ is in the ($\GL_n$) orbit of a certain circuit class $C$. In other words, decide if $f$ affine equivalent (under invertible transformations) to some polynomial in $C$.  Medini and Shpilka \cite{MS21} study the orbit of the continuant polynomial which is the trace of a certain product of matrices of dimension two and design  polynomial time reconstruction algorithms for the same.

In the context of testing whether a given polynomial promised to be in the orbit of a certain circuit class is identically zero or not, Medini and Shpilka \cite{MS21} construct hitting sets for orbits of read-once formulas\footnote{Read-once formulas are arithmetic formulas where every variable appears as a leaf at most once.} and certain {\em dense} subclasses of depth three circuits. Saha and Thankey \cite{ST24} designed hitting sets for orbits of ROABPs.  Recently, Bhargava and Ghosh \cite{BG22} obtained smaller hitting sets for the same class of polynomials. It is known \cite{ST24} that orbits of polynomial size ROABPs are a dense subclass of the the class of polynomial size \emph{general} algebraic branching programs.

Gupta et al.\ \cite{GST23} consider the equivalence problem for read-once arithmetic formulas and give a randomized polynomial-time algorithm (with oracle access to quadratic form equivalence) for the same over fields of characteristic zero. Baraskar et al.\ \cite{BDS24} give a randomized algorithm to test equivalence to {\em design polynomials}\footnote{Design polynomials are a special class of polynomials in which the degree of the GCD of every pair of monomials is bounded.} over fields of sufficiently large size and characteristic. 
In a recent, beautiful work, Baraskar et al.\ \cite{BDSS24} show $\NP$ hardness of testing equivalence to {\em sparse polynomials} over any field when the input polynomial is given in the sparse representation. The problem is as follows: given $f$ in the sparse representation and in integer $w$, does there exist an $A\in\GL_n$, $\mathbf{b}\in \mathbb{F}^n$ and a polynomial $g(\bfx)$ with at most $w$ non-zero monomials such that $f=g(A\bfx+\bfb)$? They also show that a related problem, that of deciding equivalence to constant support polynomials, is NP hard.

In this work, we consider the problem of testing equivalence to ROABPs. 

\begin{problem}[{\sf ROABP-Equivalence}]
\label{prob:orbit}
Given an $n$-variate polynomial $f\in \mathbb{F}[x_1 \ldots, x_n]$ in its sparse representation and an integer $w$ in binary, decide if there exists an $A\in GL_n(\mathbb{F})$ and a $\bfb\in\mathbb{F}^n$ such that $f(A\mathbf{x}+\mathbf{b})$ has an ROABP of width at most $w$.
\end{problem}

We show that this problem is $\NP$-hard over all fields. Over fields of characteristic $0$, it remains hard even when the input polynomial is homogeneous. To the best of our knowledge, this provides the first $\NP$-hardness result for membership testing for a dense subclass of polynomial size ABPs:

\begin{theorem}
Over fields of characteristic $0$, the {\sf ROABP-Equivalence} problem  is $\NP$-hard even when the input polynomial $f$ is homogeneous. Over fields of prime characteristic, the {\sf ROABP-Equivalence} problem is $\NP$-hard. 
\end{theorem}

Our proof combines ideas from the papers of Baraskar et al.\ \cite{BDSS24}  and Bhargava et al.\ \cite{BDGT24}. Specifically, we reduce from cutwidth and construct a polynomial $f$ such that the linear transformation which minimizes the ROABP width of $f(A\bfx + \bfb)$ is always the product of a diagonal matrix and a permutation matrix.

In our final result, we pinpoint the hardness of the original order-finding problem to the simplest non-trivial class of polynomials. We show that ROABP order finding is NP-hard even when the input is restricted to be a quadratic form (Theorem~\ref{thm:quadratic_hardness}). 

\begin{theorem}
    The problem \textrm{DenseROWidth-2} is NP-hard over all fields.
\end{theorem}

Previously, $\NP$-hardness was known for polynomials of degree $\geq6$ \cite{BDGT24}. To show this, we reduce from a different, more algebraic linear arrangement problem called linear rank-width (Problem \ref{problem:linrankwidth}).

\section{Preliminaries}
\label{sec:prelims}
We now formally define the concepts central to our results, including ROABPs, Nisan's characterization, and the graph theoretic computational problems we reduce from.

\subsection{ROABPs and Width Characterization}

\begin{definition}[Read-Once Oblivious ABP (ROABP)]
Let $\F$ be a field. An ROABP $R$ computing an $n$-variate polynomial $f(x_1, \dots, x_n)$ over $\mathbb{F}$ in a variable order $\sigma \in S_n$ is a layered, directed graph with $n+1$ layers, indexed $0$ to $n$.
\begin{itemize}
    \item The $0^{th}$ layer contains a single source vertex $s$, and the $n^{th}$ layer contains a single sink vertex $t$.
    \item Edges only exist between adjacent layers, from layer $i-1$ to layer $i$ for $i \in [n]$.
    \item Edges between layer $i-1$ and $i$ are labeled with univariate polynomials in the variable $x_{\sigma(i)}$.
\end{itemize}
The polynomial computed by the ROABP is the sum of products of edge weights over all paths from $s$ to $t$.
\end{definition}

\begin{definition}[Width of an ROABP]
The \emph{width} of an ROABP is the maximum number of vertices in any of its layers. For a polynomial $f$ and an order $\sigma$, the ROABP-width, $\rowidth_\sigma(f)$, is the width of the minimal-width ROABP for $f$ in order $\sigma$. Then $\rowidth(f)$ is defined as $\min_{\sigma \in S_n} \rowidth_\sigma(f)$.
\end{definition}

Nisan's work \cite{Nis91} provides an exact algebraic characterization of ROABP width, as the rank of a certain matrix of coefficients.

\begin{definition}[Nisan Matrix]
For a polynomial $f(x_1,\ldots,x_n) \in \F[x_1,\ldots,x_n]$ and a set of variables $X_T = \{x_i\}_{i \in T}$, the \emph{Nisan Matrix} $M_T(f)$ is a matrix whose rows are indexed by monomials in variables from $X_T$ and columns by monomials in variables from $\{x_1,\ldots,x_n\} \setminus X_T$. The $(m_1, m_2)$-th entry is the coefficient of $m_1 \cdot m_2$ in $f$.
\end{definition}

\begin{theorem}[Nisan's Characterization \cite{Nis91}]
\label{thm:nisan}
For any polynomial $f\in \mathbb{F}[x_1,\ldots,x_n]$ and order $\sigma$, the number of vertices in the $i$th layer of an optimal ROABP for $f$ in order $\sigma$ is exactly $\mathrm{rank}(M_{T_i}(f))$, where $T_i = \{\sigma(1), \dots, \sigma(i)\}$.
\end{theorem}

Next, we define two graph layout problems. They will be central to our reductions.

\subsection{Graph Layout Problems}

The authors of \cite{BDGT24} prove that DenseROwidth-d is NP hard (for $d \geq 6$) via a reduction from a particular NP-hard graph layout problem called Cutwidth. We define this problem next.

\begin{definition}[Cutwidth \cite{DPS02}]\label{def:cutwidth}
Given a graph $G=([n],E)$, a \emph{linear arrangement} is a permutation $\pi: [n] \to [n]$. The \emph{cutwidth} of $G$ with respect to $\pi$ is $\max_{i \in [n-1]} |\cut(\{\pi(1), \ldots, \pi(i)\})|$\footnote{For a graph $G=([n], E)$ and a subset $S\subseteq[n]$ of vertices, $\cut(S)$ denotes the number of edges with one endpoint in $S$ and the other in $[n]\setminus S$}. The \emph{Cutwidth} of $G$ is the minimum cutwidth over all arrangements.
\end{definition}

It is known \cite{DPS02} that the following problem is NP hard, even for graphs with maximum degree $\leq 3$:

\begin{problem}[\textrm{CutWidth} \cite{DPS02}]
Given a graph $G=([n], E)$ with maximum degree $\leq 3$ and an integer $w\in \mathbb{N}$ in binary, decide whether the cutwidfth of $G$ is at most $w$.
\end{problem}

Next we define a similar looking graph layout problem, called Linear Rank-Width. Here, the cut size is replaced by the rank of a certain matrix.

\begin{definition}[Linear Rank-Width$_{\mathbb{F}}$ \cite{Oum17}]
Let $G=([n],E)$ be a graph and let $\pi: [n] \to [n]$ be a linear arrangement. For $i \in [n-1]$, let $A_i$ be a matrix over $\mathbb{F}$ with rows indexed by vertices $\{v \mid \pi(v) \le i\}$ and columns by vertices $\{v \mid \pi(v) > i\}$. The entry $(u,v)$ is 1 if $\{u,v\} \in E$ and 0 otherwise. The \emph{linear rank-width} of $G$ with respect to $\pi$ is $\max_{i \in [n-1]} \mathrm{rank}(A_i)$. The linear rank-width of $G$ is the minimum linear rank-width over all arrangements.
\end{definition}

It is known \cite{Oum17} that the the analogous decision problem of minimizing linear rank-width is also NP-hard, over any field.

\begin{problem}[Linear Rank-Width$_{\mathbb{F}}$ \cite{DPS02}]\label{problem:linrankwidth}
Given a graph $G=([n], E)$ and an integer $w\in \mathbb{N}$ in binary, decide whether the linear rank-width of $G$ over $\mathbb{F}$ is at most $w$.
\end{problem}

\section{Inapproximability of ROABP Order-Finding}

In this section, we show the NP-hardness of approximating ROABP width up to an arbitrary constant factor, when the input is a circuit. We will require the following Lemma of Bhargava et al.\  

\begin{lemma}[\cite{BDGT24}]\label{lem:tensor}
    Given a polynomial $f(x_1, \ldots,x_n)\in\mathbb{F}[x_1, \ldots, x_n]$ with individual degree $d$ and an $l\in \mathbb{N}$, define $$f^{\otimes l}(x_1, \ldots, x_n) = \prod_{k=0}^{l-1}f\left(x_1^{(d+1)^k}, \ldots, x_n^{(d+1)^k} \right)$$ For every subset $S\subseteq[n]$, we have $\rank(M_S(f^{\otimes l})) = \rank(M_S(f))^l$. \hfill $\square$
\end{lemma}

\begin{theorem}[Inapproximability of Search-CktROwidth]
\label{thm:inapproximability}
    Let $\alpha\in \mathbb{R}$ be an arbitrary constant. Let $n,d\in \mathbb{N}$ be given as input, in unary. Let $f\in \mathbb{F}[x_1, \ldots, x_n]$ of degree $d$ be given as input, as an arithmetic circuit. Let $w\in\mathbb{N}$ also be give as input, in binary. It is NP-hard to distinguish between the following two cases:
    \begin{enumerate}
        \item $f$ has an ROABP of width $\leq w$ in some order.
        \item Every ROABP for $f$ has width $>\alpha\cdot w$.
    \end{enumerate}
\end{theorem}

\begin{proof}
    We reduce from \textrm{CutWidth} for graphs with maximum degree $\leq 3$. Given a graph $G = ([n], E)$, we first construct a small $\Pi\Sigma\Pi$ circuit computing a polynomial $f_G(x_1, \ldots, x_n)$ such that for any $S\subseteq[n]$, $\rank(M_S(f_G)) = 2^{|\cut(S)|}$. For a vertex $i\in[n]$ and a neighbour $j$ of $i$, let $n_i(j)$ denote the number of neighbours of $i$ less than or equal to $j$.

    Define $$f_G(x_1, \ldots, x_n) = \prod_{\{i,j\}\in E}\left(1+x_i^{n_i(j)}x_j^{n_j(i)}\right) = \sum_{T\subseteq E}\left(\prod_{\{i,j\}\in T}x_i^{n_{i}(j)}x_j^{n_j(i)}\right)$$

    Let $S\subseteq [n]$ be arbitrary. Define the following sets of edges: $E_1=\{\{i,j\}\mid\{i,j\}\in E \text{ and } i,j \in S\}$, $E_2 = \{\{i,j\}\mid\{i,j\}\in E \text{ and } i,j \in [n]\setminus S\}$ and $E_3 = \cut(S) = E \setminus(E_1\cup E_2)$. Then, the non-zero rows of $M_S(f_G)$ are indexed by monomials $m$ of the following type: $m$ is characterized by a subset $E_1'\subseteq E_1$ and a subset $E_3'\subseteq E_3$ such that  $m=\left(\displaystyle\prod_{\{i,j\}\in E'_1}x_i^{n_{i}(j)}x_j^{n_j(i)}\right)\left(\displaystyle\prod_{\substack{\{i,j\}\in E'_3 \\ i \in S}}x_i^{n_i(j)}\right)$. Call $E_3'$ the subset of cut edges \emph{picked} by such a monomial/row. Similarly, the non-zero columns of $M_S(f_G)$ are indexed by monomials $m'$ which are characterized by a subset $E_2'\subseteq E_2$ and a subset $E_3'\subseteq E_3$ such that we have $m'=\left(\displaystyle\prod_{\{i,j\}\in E'_2}x_i^{n_i(j)}x_j^{n_j(i)}\right)\left(\displaystyle\prod_{\substack{\{i,j\}\in E'_3 \\ j \in [n]\setminus S}}x_j^{n_j(i)}\right)$.

    Observe the following: 
    \begin{itemize}
        \item For a row indexed by $m$ and column indexed by $m'$, if the subset of cut edges picked by $m$ and $m'$ are not identical, $M_S(f_G)(m,m') = 0$.
        \item The submatrix induced by row and column monomials that pick the same subset of cut edges has rank $1$, since every row in this matrix is all $1$'s.
    \end{itemize}

    Therefore $M_S(f_G)$ is a block diagonal matrix with $2^{|\cut(S)|}$ disjoint rank $1$ blocks and so $\rank(M_S(f_G)) = 2^{|\cut(S)|}$. Next, set $l=\lceil\log \alpha\rceil+1$ and consider the polynomial $f^{\otimes l}_G$. It has degree $\leq (6|E|+1)^{\lceil\log \alpha\rceil+1}-1$, it has a formula of size $O_{\alpha}(|E|)$ which can be computed from $G$ in polynomial time, and by Lemma \ref{lem:tensor}, for every $S\subseteq[n]$ it satisfies $\rank(M_S(f_G^{\otimes l}))=2^{l\cdot|\cut(S)|}$. Therefore, it holds that \begin{itemize}
        \item If $G$ has cutwidth $\leq k$, then $\rowidth(f_G^{\otimes l})\leq 2^{l\cdot k}$
        \item If $G$ has cutwidth $\geq k+1$, then $\rowidth(f_G^{\otimes l})\geq 2^{l}\cdot2^{l\cdot k} > \alpha 2^{l\cdot k}$.
    \end{itemize}
    This finishes the proof of Theorem \ref{thm:inapproximability}.
\end{proof}

\section{Hardness of Equivalence Testing for ROABPs}
 Bhargava et al.\ \cite{BDGT24} demonstrate $\NP$-hardness of order-finding problem for ROABPs via a reduction from the cutwidth problem for graphs. More precisely, given a graph $G=([n], E)$ with maximum degree $\Delta$, the authors of \cite{BDGT24} construct the polynomial $f_G\in\mathbb{F}[x_1,\ldots, x_n]$ defined as 
\begin{equation}
\label{eq:hardness}
f_G = \sum_{\{i,j\}\in E} x_{i}^{n_i(j)}x_j^{n_j(i)}+\sum_{i=1}^{n}x_i^{\Delta+1}
\end{equation} 
where $n_i(j)\in[\Delta]$ is the number of neighbours of $i$ less than or equal to $j$. The claim (\cite{BDGT24}, Claim $4.4$) central to their reduction is that for every $S\subseteq[n]$, $\rank(M_S(f_G)) = |\cut(S)|+2$. This implies that the  ROABP width of the polynomial $f_G$ is exactly two more the cutwidth of the graph $G$ and the reduction is {\em order preserving}, i.e., an optimal  arrangement of the vertices in $G$ is exactly an optimal order for an ROABP computing $f_G$.

In this section, we prove that over all fields, testing equivalence to width $w$ ROABPs is $\NP$-hard. Over fields of characteristic $0$, we show that this problem is $\NP$-hard even when the input polynomial is \emph{homogeneous}, whereas in positive characteristic, we require inhomogeneity. Our reductions build on the reductions in Bhargava et al.\ \cite{BDGT24} and Baraskar et al.\ \cite{BDSS24}. In particular, given a graph $G=([n], E)$, we construct a polynomial $f_G \in \mathbb{F}[x_1,\ldots ,x_n
]$ with the following two properties:
\begin{enumerate}
    \item\label{item1} If $A\in\GL_n(\mathbb{F})$ is a permutation matrix times a diagonal matrix, then for every $\mathbf{b}\in\mathbb{F}^n$, the ROABP width of $f(A\bfx+\mathbf{b})$ is $\cutwidth(G)+2$.
    \item\label{item2} If $A\in \GL_n(\mathbb{F})$ is not a permutation matrix times a diagonal matrix, then for every $\mathbf{b}\in\mathbb{F}^n$, the ROABP width of $f(A\mathbf{x}+\mathbf{b})$ is at least $|E| + 3$ in \emph{every} order.
\end{enumerate}

Property \ref{item2} essentially forces the width minimizing $A$ to have a nice form, namely it is a permutation matrix times a diagonal matrix. This is because $\cutwidth(G)\leq |E|$. With this overall plan, we proceed with the details of the reduction.

\subsection{Characteristic Zero}

Over characteristic $0$, we construct a \emph{homogeneous} $f_G$. In order to show property $\ref{item2}$ for a homogeneous $f_G$, we need the following Lemma from \cite{BDSS24} that provides a lower bound on the sparsity of a polynomial divisible by a high power of a linear form with support at least $2$.

\begin{lemma}[\cite{BDSS24}]\label{lem:support}
    Let $\mathbb{F}$ be a field of characteristic $0$. Let $l\in\mathbb{F}[x_1, \ldots, x_n]$ be a linear polynomial with support $\geq 2$, let $h\in\mathbb{F}[x_1, \ldots, x_n]$ be arbitrary and let $d\in\mathbb{N}$. Then $l^d\times h$ has at least $d+1$ non-zero monomials.  
\end{lemma}

In fact, we need a strengthening of this lemma. In the sequel, we strengthen this lemma and show that any polynomial divisible by the $d$th power of a linear form with support at least $2$ has \emph{ROABP width} at least $d+1$, in every order.

\begin{lemma}
\label{lem:rank}
Let $\mathbb{F}$ be a field of characteristic zero. Let $\ell$ be a linear polynomial with support at least $2$ and $h\in \mathbb{F}[x_1, \ldots ,x_n]$ be any non-zero polynomial and $d\in \mathbb{N}$. Let $F= \ell^d \cdot h$. Then for any $\sigma \in S_n$, $\rowidth_{\sigma}(f) \geq d+1$. 
\end{lemma}
\begin{proof}
Without loss of generality, let $x_1, x_2$ be in the support of $\ell$, i.e., $\ell = a_1x_1 + a_2x_2 + \ell'$ such that $a_1, a_2 \neq 0$. Then, $F =\ell^d \cdot h =  (a_1x_1 + a_2x_2 + \ell')^d \cdot h$. We can view $h$ as a polynomial in $\mathbb{F}(x_3, \ldots, x_n)[x_1, x_2]$, i.e., a polynomial in $x_1, x_2$ with coefficients in the field $\mathbb{F}(x_3, \ldots, x_n)$. Let $k$ be the degree of $h$ (over $\mathbb{F}(x_3, \ldots, x_n)$). That is, $k$ is the maximal $i+j$ such that $h$ contains the monomial $c_{i,j}x_1^ix_2^j$ with $c_{i,j}\in \mathbb{F}(x_3, \ldots, x_n)$ and $ c_{i,j}\neq 0$. We denote by $h_k$ the homogeneous component of $h$ of degree $k$. Then, $F = F_1 + F_2$ where $F_1 = (a_1x_1 + a_2x_2)^d \cdot h_k$. Every monomial $m$ in $F_1$ satisfies $\deg_{x_1}(m)+\deg_{x_2}(m)= d+k$. Furthermore, every monomial in $F_2$ satisfies $\deg_{x_1}(m)+\deg_{x_2}(m) < k +d$. Next, by applying Lemma \ref{lem:support} to $(a_1x_1+a_2x_2)^{d}\times h_k$ \emph{over the field $\mathbb{F}(x_3, \ldots, x_n)$},  we get  $(a_1x_1+a_2x_2)^{d}\times h_k  =\sum\limits_{\substack{i,j \\ i+j=d+k}}c_{i,j}x_1^{i}x_2^j$ such that $c_{i,j}\in \mathbb{F}[x_3, \ldots, x_n]$ and at least $d+1$ of the $c_{i,j}$'s non-zero. Let $m_{i,j}$ be the leading monomial of $c_{i,j}$. Observe that:
\begin{itemize}
\item There exists a set $P$ of at least $d+1$ pairs $(i,j)$ with $i+j=d+k$ such that for each pair $(i,j)$ in $P$, the coefficient of the monomial $x_1^ix_2^jm_{i,j}$ in $F_1$ is non-zero. 
\item There is no monomial $m$ in $F$ with non-zero coefficient such that $\deg_{x_1}(m)+\deg_{x_2}(m)>d+k$. 
\end{itemize}
Now, consider any ROABP for $F$, in an arbitrary order $x_{\pi(1)}, \ldots, x_{\pi(n)}$. Let $S$ be a {\em prefix} of $\pi$ that separates $1$ and $2$. Without loss of generality, assume that $1\in S$ and $2\in [n]\setminus S$. Since width of the ROABP is at least $\rank(M_S)(F)$ we now prove that $\rank(M_S(F)) \geq d+1$. 
For a monomial $m\in \mathbb{F}[x_1, \ldots, x_n]$, let $m[S]$ denote the monomial obtained from $m$ by setting variables outside $\{x_i\mid i\in S\}$ to $1$. Consider the submatrix of $M_S(F)$ induced by the row monomials $\{m_{i,j}[S]x_1^i\mid (i,j)\in P\}$ (order the rows by increasing values of $i$) and the column monomials $\{m_{i,j}[[n]\setminus S]x_2^j\mid (i,j) \in P\}$ (order the columns by increasing value of $j$). By our observations above, this is a full rank, square, anti-triangular matrix with non-zero entries on the main anti-diagonal. It has at least $d+1$ rows.  Therefore, the rank of this submatrix is at least $|P|\geq d+1$ and so is the rank of $M_S(F)$.
\end{proof}

\begin{theorem}
\label{thm:orbit_hardness}
The Equivalence to ROABP problem (Problem~\ref{prob:orbit}) is $\NP$-hard over fields of characteristic $0$. $\NP$-hardness holds even when the input polynomial $f$ is homogeneous.
\end{theorem}

\begin{proof}
    We reduce from \textrm{CutWidth}. Let $(G=([n], E), w)$ be an instance of \textrm{CutWidth}. We map it to an instance $(f_G, w+2)$ of \textrm{ROABP-Equivalence}. To this end, introduce a total order $e_1<e_2<\ldots<e_{|E|}$ on the edges of $G$. For $i \in E$, let $e_i=\{i_1,i_2\}$ such that $i_1<i_2$. Let $\mathbf{x} = (x_1, \ldots, x_n)$. For $S\subseteq[n]$, define $\Pi_S(\mathbf{x}) \triangleq \left(\prod_{j\in S}x_j^{|E|+2}\right)$ and define the polynomial $$f_G(\mathbf{x}) \triangleq \Pi_{[n]}(\mathbf{x}) \left(\sum_{i=1}^{|E|}x_{i_1}^{i}x_{i_2}^{2|E|-i+1}\right)$$

Clearly, $f_G$ is a homogeneous polynomial of degree $(n+2)|E|+2n+1$. We now prove both the forward and reverse directions of the reduction.  \\

\noindent \textbf{Forward Direction: If $G$ has cutwidth at most $w$, then there exist $A\in \GL_n(\mathbb{F})$ and $\mathbf{b}\in \mathbb{F}^n$ such that $f_G(A\mathbf{x}+\mathbf{b})$ has an ROABP of width at most $w+2$ in some order. }
    
The proof of the forward direction follows the outline of the corresponding proof in \cite{BDGT24}. In particular, we pick $A=I_n$, the $n\times n$ identity matrix, and $\mathbf{b}=0$. Consider any subset $S\subseteq[n]$ and $f_G$ ca be expressed as 
    
    \begin{align*}
        f_G(\mathbf{x}) = {} &\underbrace{\left(\Pi_S(\mathbf{x})\sum_{\substack{i\in[|E|] \\ i_1, i_2\in S}}x_{i_1}^{i}x_{i_2}^{2|E|-i+1}\right)\Pi_{[n]\setminus S}(\mathbf{x})}_{f_1} \\
        & + \underbrace{\left(\Pi_{[n]\setminus S}(\mathbf{x})\sum_{\substack{i\in[|E|] \\ i_1, i_2\in [n]\setminus S}}x_{i_1}^{i}x_{i_2}^{2|E|-i+1}\right)\Pi_{S}(\mathbf{x})}_{f_2} \\
        & + \underbrace{\left(\sum_{\substack{i\in [|E|] \\ \{i_1, i_2\}\in \cut(S)}}\Pi_{S}(\mathbf{x})\Pi_{[n]\setminus S}(\mathbf{x})x_{i_1}^{i}x_{i_2}^{2|E|-i+1}\right)}_{f_3}
    \end{align*}

Observe that $f_1$ and $f_2$ are non-zero polynomials of the form $g(\bfx_S)\times h(\bfx_{[n]\setminus S})$, where for a subset $S$ of $[n]$, $\bfx_{S}$ are the $x$ variables indexed by $S$. Therefore, we have that $\rank(M_S(f_1))=\rank(M_S(f_2)) = 1$. Finally, notice that $f_3$ can be written as a \emph{sum} of $|\cut(S)|$ monomials. For every monomial $m$, $\rank(M_S(m)) = 1$ for each $S$. Since $M_S(f_G) = M_S(f_1) + M_S(f_2) + M_S(f_3)$, by subadditivity of rank, we have that $\rank(M_S(f_G)) \leq |\cut(S)|+2$. Consider a linear arrangement $\pi:[n]\rightarrow[n]$ that witnesses $\cutwidth(G)\leq w$. Due to the above reasoning, combined with Nisan's characterization, we have that $f_G$ has an ROABP in order $\pi$, of width $\leq w+2$. \\

\noindent\textbf{Reverse Direction: If there exist $A\in \GL_n(\mathbb{F})$ and $\mathbf{b}\in \mathbb{F}^n$ such that $f_G(A\mathbf{x}+\mathbf{b})$ has an ROABP of width at most $w+2$ in some order, then $G$ has cutwidth at most $w$.}

    In order to prove the reverse direction, we need the following key lemma.

    \begin{lemma}\label{lem:width}
        Suppose $A\in\GL_n(\mathbb{F})$ is not the product of a permutation matrix and a diagonal matrix, then for \emph{every} $\mathbf{b}\in\mathbb{F}^n$, every ROABP computing $f_G(A\bfx+\mathbf{b})$ must have width at least $|E|+3$.
    \end{lemma}
    \begin{proof}
        If $A$ is not the product of a permutation matrix and a diagonal matrix, then $A\bfx+\mathbf{b}$ must send at least one $x$ variable to a linear polynomial with support \emph{at least} $2$. We may then write $f_G(A\bfx+\bfb) = l(\bfx)^{|E|+2}\times h$ for a nonzero polynomial $h$. By Lemma \ref{lem:rank}, for any $\sigma \in S_n$, $\rowidth_{\sigma}(f_G(A\bfx+\mathbf{b})) \geq |E|+3$. 
    \end{proof}
We use Lemma \ref{lem:width} to complete the reverse direction of the reduction. First, we observe that if $A$ is the product of a diagonal matrix and a permutation matrix, then $\rowidth_{\sigma}(f_G(A\bfx+b))=\rowidth_{\sigma}(f_G(\bfx))$: Suppose $A$ has this form. Then there exist $a_1, \ldots, a_n\in\mathbb{F}$, all non-zero, and a permutation $\pi:[n]\rightarrow [n]$ such that $A\bfx + \bfb = (a_1x_{\pi(1)}+b_1, a_2x_{\pi(2)}+b_2, \ldots, a_nx_{\pi(n)}+b_n)$. 
We can obtain an ROABP for $f_G(\bfx)$ from an ROABP for $f_G(a_1x_{\pi(1)}+b_1, \ldots, a_nx_{\pi(n)}+b_n)$ by replacing each $x_{\pi(i)}$ with $(x_i-b_i)/a_i$. The resulting ABP is still an ROABP, with the same width as before. Also, this process is clearly reversible.  \\

By the proof of the forward direction, we know that $\rowidth(f_G)\leq \cutwidth(G)+2\leq |E|+2$. 
Now suppose there exist $A\in \GL_n$ and $\bfb\in\mathbb{F}^n$ such that $f_G(A\bfx+\bfb)$ has ROABP width at most $w+2$ in some order. Due to Lemma \ref{lem:width} and the observation above, we may assume that $A$ is the identity matrix and $\bfb=\mathbf{0}$. Next, we show that for each $S\subseteq[n]$, $\rank(M_S(f_G))\geq |\cut(S)|+2$ by a proof similar to the $\NP$ hardness in \cite{BDGT24} by exhibiting a submatrix of $M_S(f_G)$ that is a $|\cut(S)|\times |\cut(S)|$ permutation matrix, along with two rows that lie in disjoint spaces.

 This suffices for the reverse direction, for if $f_G$ has an ROABP of width $\leq w+2$ in order $\pi$, then due to Nisan's characterization (Theorem \ref{thm:nisan}), we would have that $\cutwidth(G)\leq w$, witnessed by the linear arrangement $\pi$. 
    
    Define $E_1\triangleq\{e_i\mid e_i\in \cut(S), i_1\in S, i_2\in [n]\setminus S\}$ and $E_2\triangleq \{e_i\mid e_i\in \cut(S), i_1\in [n]\setminus S, i_2\in S\}$. We look at the submatrix of $M_S(f_G)$ induced by the row monomials $R=\{\Pi_S(\bfx)x_{i_1}^{i}\mid e_i\in E_1\}\cup\{\Pi_S(\bfx)x_{i_2}^{2|E|+1-i}\mid e_i\in E_2\}$ and column monomials $C=\{\Pi_{[n]\setminus S}(\bfx)x_{i_1}^{i}\mid e_i\in E_2\}\cup\{\Pi_{[n]\setminus S}(\bfx)x_{i_2}^{2|E|+1}\mid e_i\in E_1\}$. This is a permutation matrix, since each monomoial labeling both the rows an columns can be associated with a unique end point of an edge in $\cut(S)$, and the only non-zero entry (in the entirety of $M_S(f_G)$, not just the submatrix) in that row/column corresponds to the monomial labeling the column/row associated with the other end point of that edge. On the other hand, consider the row labeled by $\Pi_S(\bfx)$. This has non-zero entries in the columns labeled by the monomials $\Pi_{[n]\setminus S}(\bfx)x_i^{2|E|+1}$, for each $i\in [n]\setminus S$ (note that these monomials are \emph{not} contained in $C$), and a zero entry in the column labeled by $x_{[n]\setminus S}$. Also, a row labeled by $\Pi_S(\bfx)x_{i}^{2|E|+1}$ for an $i\in[S]$ (again, note that this monomial is \emph{not} contained in $R$) has a non-zero entry in the column labeled by $\Pi_{[n]\setminus S}$. This gives us that $\rank(M_S(f_G))\geq |\cut(S)|+2$.\end{proof}

\subsection{Characteristic $p$}

In this section, we prove hardness of testing equivalence to width $w$ ROABPs, over characteristic $p$. In this setting, we resort to inhomogeneity to prove hardness. This is because Lemma \ref{lem:support}, and therefore, Lemma \ref{lem:width} fail to hold over small characteristic. In particular, we show that a skewed version of the polynomial in Equation (\ref{eq:hardness}) gives us hardness, even over characteristic $p$. Recall the polynomial $f_G = \sum_{\{i,j\}\in E} x_{i}^{n_i(j)}x_j^{n_j(i)}+\sum_{i=1}^{n}x_i^{\Delta+1}$ constructed in \cite{BDGT24}. Instead of the term $\sum_{i=1}^{n}x_i^{\Delta+1}$, we introduce the asymmetric $\sum_{i=1}^{n}x_i^{D_j}$. We carefully choose distinct exponents $D_j$ such that they are polynomial in the size of the graph while also allowing us to prove NP-hardness.

We will use the following well known result of Lucas.

\begin{theorem}\label{thm:lucas}
    Let $p$ be a prime and $m,n$ be integers such that $m=\sum_{k=0}^{t}m_kp^k$ and $n=\sum_{k=0}^{t}n_kp^k$ are the base $p$ expansions of $m$ and $n$ respectively. Then, $$\binom{m}{n}\equiv\prod_{k=0}^{t}\binom{m_k}{n_k} \mod p$$ where we use the convention that $\binom{m_k}{n_k}=0$ if $m_k<n_k$.
\end{theorem}

\begin{theorem}
\label{thm:orbit_hardness_p}
Let $p$ be a prime and let $\mathbb{F}$ be a field of characteristic $p$. The Equivalence to ROABP problem (Problem~\ref{prob:orbit}) is $\NP$-hard over $\mathbb{F}$.
\end{theorem}

\begin{proof}
Let $(G=([n], E), w)$ be an instance of the \textrm{CutWidth} problem with maximum degree of $G$ is $\leq 3$. We map it to an instance $(f_G, w+2)$ of ROABP-Equivalence.

Let $M = \max\{|E|+4,7\}$. Let $L$ be the smallest integer such that $p^L > M$. For each $j \in [n]$, define the exponent
    $$ D_j := (p^L - 1) + (j-1)p^L. $$
    Note that each $D_j$ is $\poly(n, |E|)$. For each vertex $i$, let $n_i(j)\in[\deg(i)]$, as before, be number of neighbours of $i$ less than or equal to $j$. Note that $n_{i}(j)\leq 3$ for each $\{i,j\} \in E$. Let $\bfx = (x_1, \ldots, x_n)$. Define the polynomial
\begin{equation}
\label{eq:charp}
f_G(\bfx)\triangleq\sum_{\{i,j\}\in E}x_{i}^{n_i(j)}x_j^{n_j(i)}+\sum_{j=1}^{n}x_j^{D_j}.
\end{equation}

First, we check that a statement analogous to Claim 4.4 in \cite{BDGT24} continues to hold for the $f_G$ we have defined. The proof in \cite{BDGT24} works for our $f_G$ as well, we include a proof here for completeness.

    \begin{claim}[Analogous to Claim 4.4 in \cite{BDGT24}]\label{claim:rank}
        For each $S\subseteq [n]$, we have $\rank(M_S(f_G)) = |\cut(S)|+2$.
    \end{claim}

    \begin{proof}
The nonzero rows of $M_S(f_G)$ are of at most three types: 
\begin{enumerate}
\item[(1)] Row indexed by the constant monomial $1$; 
\item[(2)] Rows indexed by $x_i^{n_i(j)}$ for $i\in S$ and $j\in [n]\setminus S$;
\item[(3)] Rows indexed by $x_i^{n_i(j)}x_j^{n_j(i)}$ for $i,j\in S$ or by $x_i^{D_i}$ for $i\in S$.
\end{enumerate}
Similarly, columns of $M_S(f_G)$ are of at most three types:
\begin{itemize}
\item[(1)] Column indexed by  the constant monomial $1$;
\item[(2)] Columns indexed by $x_j^{n_j(i)}$ for $j\in [n]\setminus$ and $i\in S$;
\item[(3)] Columns indexed by $x_i^{n_i(j)}x_j^{n_j(i)}$ for $i,j\in [n]\setminus S$ or by $x_j^{Dj}$ for $j\in [n]\setminus S$.
\end{itemize} 
Let $M_{i,j}$ denote the submatrix of $M_S(f_G)$ induced by rows of type $(i)$ and columns of type $(j)$ from the possibilities mentioned above. By construction, $M_{2,2}$ is a $|\cut(S)|\times |\cut(S)|$ permutation matrix. On the other hand, $M_{1,3}$ and $M_{3,1}$ are non-zero row and column matrices respectively (due to the presence of the $\sum x_j^{D_j}$ monomials in $f_G$). All the other submatrices $M_{i,j}$ are zero matrices. Therefore, $\rank(M_S(f_G)) = \rank(M_{2,2})+\rank(M_{1,3})+\rank(M_{3,1}) = |\cut(S)|+2$.
\end{proof}
\textbf{Forward Direction: If $G$ has cutwidth at most $w$, then there exist $A\in \GL_n(\F)$ and $\bfb\in \F^n$ such that $f_G(A\bfx+\bfb)$ has an ROABP of width at most $w+2$ in some order.} 
    
    We pick $A=I_n$ and $\bfb=0$. By Claim \ref{claim:rank}, we have for each $S\subseteq [n]$, that $\rank(M_S(f_G)) = |\cut(S)|+2$. If there is a linear arrangement $\pi$ that witnesses $\cutwidth(G)\leq w$, then by Nisan's characterization (Theorem \ref{thm:nisan}), $f_G$ has an ROABP in order $\pi$ of width $\leq w+2$.

    \textbf{Reverse Direction: If there exist $A\in \GL_n(\F)$ and $\bfb\in \F^n$ such that $f_G(A\bfx+\bfb)$ has an ROABP of width at most $w+2$ in some order, then $G$ has cutwidth at most $w$.}

We first prove the following key lemma which is the analogue of Lemma \ref{lem:width} in the case of characteristic $p$:

    \begin{lemma}\label{lem:width_char_p}
        Suppose $A\in\GL_n(\F)$ is not the product of a permutation matrix and a diagonal matrix. Then for \emph{every} $\bfb\in\F^n$, every ROABP computing $f_G(A\bfx+\bfb)$ must have width at least $|E|+3$.
    \end{lemma}
    \begin{proof}
        If $A$ is not the product of a permutation matrix and a diagonal matrix, there is a row of $A$, say $A_j$, with at least two non-zero entries. Let $j$ be the largest index for which this holds. The linear form $l_j(\bfx) = A_j\bfx + b_j$ has support at least 2. Assume, without loss of generality, that $l_j(\bfx) = a_{jk_1}x_{k_1} + a_{jk_2}x_{k_2} + \dots$ with $k_1 \neq k_2$ and $a_{jk_1}, a_{jk_2} \neq 0$.

        The polynomial $f_G(A\bfx+\bfb)$ contains the term $(l_j(\bfx))^{D_j}$. By our choice of $D_j = (p^L - 1) + (j-1)p^L$, Lucas's Theorem (Theorem \ref{thm:lucas}) guarantees that $\binom{D_j}{i} \not\equiv 0 \pmod p$ for all $1 \le i \le M<D_j$. This ensures that in the expansion of $(l_j(\bfx))^{D_j}$, the coefficient of $x_{k_1}^i x_{k_2}^{D_j-i}$ is non-zero for all $1 \le i \le M-1$. Note that the support of these monomials is at least $2$.
        
        The other terms in $f_G(A\bfx+\bfb)$ have lower total degree or have support at most $1$: For $l > j$, the term $(A_l\bfx+b_l)^{D_l}$ only involves one variable. For $l < j$, the term $(A_l\bfx+b_l)^{D_l}$ has total degree $D_l < D_j$. The edge terms have total degree at most $6 < D_j$. Thus, the analysis is dominated by $(l_j(\bfx))^{D_j}$.

        Now, consider any ROABP for $f_G(A\bfx+\bfb)$ in an arbitrary order $\pi$. Pick a prefix $S$ of $\pi$ that separates $k_1$ and $k_2$. The width of the ROABP is at least $\rank(M_S(f_{G}(A\bfx+\bfb)))$. The submatrix of $M_S$ corresponding to row monomials $\{x_{k_1}^i\}_{i=1}^{M-1}$ and column monomials $\{x_{k_2}^{D_j-i}\}_{i=1}^{M-1}$ is an anti-triangular matrix of size $(M-1) \times (M-1)$ with non-zero entries on its anti-diagonal. Its rank is therefore $M-1 \geq |E|+3$. Since $\pi$ was arbitrary, this holds for any order.
    \end{proof}

    We use Lemma \ref{lem:width_char_p} to prove the reverse direction. Suppose there exist $A\in \GL_n(\F)$ and $\bfb\in\F^n$ such that $f_G(A\bfx+\bfb)$ has ROABP width at most $w+2$. Combining with Lemma \ref{lem:width_char_p} the reasoning provided in the proof of Theorem \ref{thm:orbit_hardness} we may assume that $A$ is the identity matrix and that $\mathbf{b}=\mathbf{0}$. By Claim \ref{claim:rank}, we have
    $$ \cutwidth(G)+2 = \rowidth(f_G(\bfx)) \le w+2, $$
    which implies $\cutwidth(G) \le w$. This completes the soundness argument and the proof of the theorem.
\end{proof}

\section{Hardness of Order-Finding for Quadratic Forms}

The NP hardness reduction for ROABP order finding provided by \cite{BDGT24} embeds the cut sizes of the graph into the ranks of the corresponding Nisan matrices. Instead, we can also embed the \emph{cut-rank} information into the Nisan matrices. Let $G$ be a graph and $S$ be a subset of it's vertices. For a field $\mathbb{F}$, $\cutrank_{\mathbb{F}}(S)$ is defined as the rank over $\mathbb{F}$ of the $0-1$ matrix whose rows are indexed by vertices of $G$ in $S$, columns by vertices not in $S$, and the $(u,v)$-th entry is $1$ iff $\{u,v\}$ is an edge of $G$. This leads us to our next reduction, which gives NP-Hardness of order finding for quadratic forms. This is an improvement over \cite{BDGT24}, who give hardness for degree $6$ polynomials).

\begin{theorem}
\label{thm:quadratic_hardness}
DenseROwidth-2 is NP-hard over any field $\mathbb{F}$.
\end{theorem}

\begin{proof}
    We reduce from the linear rank-width problem over $\mathbb{F}$. Let $(G = ([n], E), w)$ be an instance of the linear rank-width problem. Construct the polynomial $f_G$ over $\mathbb{F}$ defined as $$f_{G} = \sum_{\{i,j\}\in E}x_ix_j + \sum_{i=1}^{n}x_i^2$$

    The proof of NP-hardness follows from the next claim. 

    \begin{claim}
        Let $S\subseteq[n]$ be such that $1\leq|S|\leq n-1$. Then $\rank(M_S(f_G))=\cutrank_{\mathbb{F}}(S)+2$ 
    \end{claim}
    \begin{proof}
        The proof follows by inspecting the structure of the matrix $M_S(f_G)$. Since $f_G$ is a quadratic form, we only need to consider monomials of degree at most 2. We partition the rows and columns of the Nisan matrix $M_S(f_G)$ by the degree of the indexing monomials and obtain the following block structure for $M_S(f_G)$
    \[
    M_S(f_G) \cong
    \begin{pmatrix}
    M_{0,0} & M_{0,1} & M_{0,2} \\
    M_{1,0} & M_{1,1} & M_{1,2} \\
    M_{2,0} & M_{2,1} & M_{2,2}
    \end{pmatrix}
    \]
    Here, $M_{i,j}$ is the submatrix induced by row monomials of degree $i$ and column monomials of degree $j$. First, note that $M_{1,1}$ is exactly the cut-rank matrix for the subset $S$. Also, note that $M_{2,0}$ and $M_{0,2}$ are non-zero (because of the $x_i^2$ terms) column and row matrices respectively, and so they have rank $1$. Finally, note that the rest of the $M_{i,j}$ are all $\mathbf{0}$. Therefore, $\rank(M_S(f_G)) = \rank(M_{0,2})+\rank(M_{2,0})+\rank(M_{1,1}) = \cutrank_{\mathbb{F}}(S)+2$
    \end{proof}
    In particular, this implies that $\rowidth(f_G)$ is linear rank-width of $G$ plus 2.
\end{proof}


\bibliography{ref}

\begin{thebibliography}{10}

\bibitem{AMS11}
Christoph Amb\"{u}hl, Monaldo Mastrolilli, and Ola Svensson.
\newblock Inapproximability results for maximum edge biclique, minimum linear arrangement, and sparsest cut.
\newblock {\em SIAM Journal on Computing}, 40(2):567--596, 2011.
\newblock \href {https://arxiv.org/abs/https://doi.org/10.1137/080729256} {\path{arXiv:https://doi.org/10.1137/080729256}}, \href {https://doi.org/10.1137/080729256} {\path{doi:10.1137/080729256}}.

\bibitem{APW12}
Per Austrin, Toniann Pitassi, and Yu~Wu.
\newblock Inapproximability of treewidth, one-shot pebbling, and related layout problems.
\newblock In Anupam Gupta, Klaus Jansen, Jos{\'e} Rolim, and Rocco Servedio, editors, {\em Approximation, Randomization, and Combinatorial Optimization. Algorithms and Techniques}, pages 13--24, Berlin, Heidelberg, 2012. Springer Berlin Heidelberg.

\bibitem{BDSS24}
O.~Baraskar, A.~Dewan, C.~Saha, and P.~Sinha.
\newblock {NP}-hardness of testing equivalence to sparse polynomials and to constant-support polynomials.
\newblock In {\em Proceedings of the 51st International Colloquium on Automata, Languages, and Programming (ICALP)}, 2024.

\bibitem{BDS24}
Omkar Baraskar, Agrim Dewan, and Chandan Saha.
\newblock Testing equivalence to design polynomials.
\newblock In Olaf Beyersdorff, Mamadou~Moustapha Kant{\'{e}}, Orna Kupferman, and Daniel Lokshtanov, editors, {\em 41st International Symposium on Theoretical Aspects of Computer Science, {STACS} 2024, March 12-14, 2024, Clermont-Ferrand, France}, volume 289 of {\em LIPIcs}, pages 9:1--9:22. Schloss Dagstuhl - Leibniz-Zentrum f{\"{u}}r Informatik, 2024.
\newblock URL: \url{https://doi.org/10.4230/LIPIcs.STACS.2024.9}, \href {https://doi.org/10.4230/LIPICS.STACS.2024.9} {\path{doi:10.4230/LIPICS.STACS.2024.9}}.

\bibitem{BDGT24}
V.~Bhargava, P.~Dutta, S.~Ghosh, and A.~Tengse.
\newblock The complexity of order-finding for roabps.
\newblock arXiv preprint arXiv:2411.18981, 2024.
\newblock \href {https://arxiv.org/abs/2411.18981} {\path{arXiv:2411.18981}}.

\bibitem{BG22}
Vishwas Bhargava and Sumanta Ghosh.
\newblock Improved hitting set for orbit of roabps.
\newblock {\em computational complexity}, 31(2):15, Oct 2022.
\newblock \href {https://doi.org/10.1007/s00037-022-00230-9} {\path{doi:10.1007/s00037-022-00230-9}}.

\bibitem{DPS02}
J.~Díaz, J.~Petit, and M.~Serna.
\newblock A survey of graph layout problems.
\newblock {\em ACM Computing Surveys}, 34(3):313--356, 2002.

\bibitem{GST23}
Nikhil Gupta, Chandan Saha, and Bhargav Thankey.
\newblock {\em Equivalence Test for Read-Once Arithmetic Formulas}, pages 4205--4272.
\newblock URL: \url{https://epubs.siam.org/doi/abs/10.1137/1.9781611977554.ch162}, \href {https://arxiv.org/abs/https://epubs.siam.org/doi/pdf/10.1137/1.9781611977554.ch162} {\path{arXiv:https://epubs.siam.org/doi/pdf/10.1137/1.9781611977554.ch162}}, \href {https://doi.org/10.1137/1.9781611977554.ch162} {\path{doi:10.1137/1.9781611977554.ch162}}.

\bibitem{MS21}
Dori Medini and Amir Shpilka.
\newblock Hitting sets and reconstruction for dense orbits in vp{\_}\{e\} and {\(\Sigma\)}{\(\Pi\)}{\(\Sigma\)} circuits.
\newblock In Valentine Kabanets, editor, {\em 36th Computational Complexity Conference, {CCC} 2021, July 20-23, 2021, Toronto, Ontario, Canada (Virtual Conference)}, volume 200 of {\em LIPIcs}, pages 19:1--19:27. Schloss Dagstuhl - Leibniz-Zentrum f{\"{u}}r Informatik, 2021.
\newblock URL: \url{https://doi.org/10.4230/LIPIcs.CCC.2021.19}, \href {https://doi.org/10.4230/LIPICS.CCC.2021.19} {\path{doi:10.4230/LIPICS.CCC.2021.19}}.

\bibitem{Nis91}
N.~Nisan.
\newblock Lower bounds for non-commutative computation.
\newblock In {\em Proceedings of the 23rd Annual ACM Symposium on Theory of Computing (STOC)}, pages 410--418, 1991.

\bibitem{Oum17}
S.~Oum.
\newblock Rank-width: Algorithmic and structural results.
\newblock {\em Discrete Applied Mathematics}, 231:15--24, 2017.

\bibitem{ST24}
Chandan Saha and Bhargav Thankey.
\newblock Hitting sets for orbits of circuit classes and polynomial families.
\newblock {\em ACM Trans. Comput. Theory}, 16(3), September 2024.
\newblock \href {https://doi.org/10.1145/3665800} {\path{doi:10.1145/3665800}}.

\end{thebibliography}

\end{document}